%% file: main2.tex
\documentclass{article}
\usepackage[utf8]{inputenc}
\usepackage{parskip,graphicx,textcomp,amsmath,amsfonts,hyperref}
\hypersetup{colorlinks=true,citecolor=black,linkcolor=blue,urlcolor=blue}
\usepackage[numbers]{natbib}

\input{commands.tex}

\input{numbering2.tex}

\begin{document}

    \title{A point to set principle for finite-state dimension}
    \author{Elvira
Mayordomo\thanks{Departamento de Inform{\'a}tica e Ingenier{\'\i}a
de Sistemas, Instituto de Investigaci{\'o}n en Ingenier{\'\i}a de
Arag{\'o}n, Universidad de Zaragoza, Spain. Research supported in
part by Spanish Ministry of Science and Innovation grant
PID2019-104358RB-I00  and by the Science dept. of Aragon Government:
Group Reference T64$\_$20R (COSMOS research group).}}
    \date{\today}

    \maketitle

    \begin{abstract}
Effective dimension has proven very useful in geometric measure
theory through the point-to-set principle \cite{LuLu18}\ that
characterizes Hausdorff dimension by relativized effective
dimension. Finite-state dimension is the least demanding
effectivization in this context \cite{FSD}\ that among other results
can be used to characterize Borel normality \cite{BoHiVi05}.

In this paper we prove a characterization of finite-state dimension
in terms of information content of a real number at a certain
precision. We then use this characterization to give a robust concept of relativized normality and prove a finite-state
dimension point-to-set principle. We finish with an open question on the equidistribution properties of relativized normality.

    \end{abstract}

\section{Introduction}

Effective dimension was introduced in \cite{DISS,DCC}\ as an
effectivization of Hausdorff dimension. One of its generalizations
is finite-state dimension \cite{FSD}\ that is a robust notion that
interacts with compression and characterizes Borel normality
\cite{BoHiVi05}.

In \cite{LuLu18}\  Lutz and Lutz proved a point-to-set principle
that characterizes Hausdorff dimension in terms of relativized
effective dimension. This principle has already produced a number of
interesting results in geometric fractal theory through
computability based proofs (see
\cite{LuLu20,LuMa21b}\ and a number of more recent results such as \cite{FiSt24}).

In this paper we provide a characterization of finite-state
dimension on  Euclidean space  based on the finite-state
information content of a real number at a certain precision, which
also provides an alternative characterization of Borel normality. This new characterization gives rise to a natural and robust relativization of finite state dimension with the strong property of a  finite-state
dimension point-to-set principle. We finish with open questions on the equidistribution properties of the corresponding relativized normality.

\section{Preliminaries}

Let $\Sigma$ be a finite alphabet. We write $\Sigma^*$ for the set
of all (finite) {\sl strings\/} over $\Sigma$ and $\Sigma^{\infty}$
for the set of all (infinite) {\sl sequences\/} over $\Sigma$. We
write $|x|$ for the length of a string or sequence $x$, and we write
$\lambda$ for the {\sl empty string}, the string of length 0. For
$x\in \Sigma^*\cup\Sigma^{\infty}$ and $0\le n<|x|$, we write
$x\upto n=x[0..n-1]$. For $w\in\Sigma^*$ and $x\in
\Sigma^*\cup\Sigma^{\infty}$, we say that $w$ is  a {\sl prefix\/}
of $x$, and we write $w\sqsubseteq x$, if $x\upto |w|=w$.

A {\sl $\Sigma$ finite-state transducer ($\Sigma$-FST)\/} is a
4-tuple $T = (Q, \delta, \nu, q_0)$, where
\begin{itemize}
\item $Q$ is a nonempty, finite set of {\it states},
\item ${\delta}:{Q \times \Sigma}\to{Q}$ is the {\it transition
function},
\item ${\nu}:{Q \times \Sigma}\to{\Sigma^{*}}$ is the {\it output
function}, and
\item $q_0 \in Q$ is the {\it initial state}.
\end{itemize}
For $q \in Q$ and $w \in \Sigma^{*}$, we define the {\it output}
from state $q$ on input $w$ to be the string $\nu(q, w)$ defined by
the recursion
\begin{align*}
&\nu(q, \lambda) = \lambda, \\
&\nu(q, wa) = \nu(q, w)\nu(\delta(q, w), a)
\end{align*}
for all $w \in \Sigma^{*}$ and $a \in \Sigma$.  We then define the
{\it output} of $T$ on input $w \in \Sigma^{*}$ to be the string
$T(w) = \nu(q_0, w)$.

For each integer $b\ge 1$ we let $\Sigma_b=\{0, 1, \ldots, b-1\}$ be
the alphabet of {\sl base-$b$\/}  {\sl digits}. We  use infinite
sequences over $\Sigma_b$ to represent real numbers in  $[0,1)$.
Each $S\in\Sigma_b^{\infty}$ is associated the real number
$\real_b(S)=\sum\limits_{i=1}^{\infty} S[i-1] b^{-i}$ and for each
$x\in [0,1)$, $\seq_b(x)$ is the infinite sequence $S$ that does not
finish with infinitely many $b-1$ and such that $x=\real_b(S)$.

   A set of real numbers   $A \subseteq
  [0,1) $  is represented by the set $$\seq_b(A) = \{\seq_b(\alpha) \mid
  \alpha \in A\}$$ of sequences. If $X\subseteq\Sigma_b^\infty$ then
  \[\real_b(X)=\{\real_b(x) \mid
  x \in X\}.\]

  We will denote with $\D_b$ the set of rational numbers that have
  finite representation in base $b$, that is,
  \[\D_b=\myset{q}{\seq_b(q)=w0^{\infty}, w\in\Sigma_b^{*}}.\] We will
  write $\real_b(w)$ for $\real_b(w0^{\infty})$ when
  $w\in\Sigma_b^{*}$.

  \section{An euclidean characterization of finite-state dimension and Borel normality}

  Finite-state dimension was introduced in \cite{FSD}\ on the space of infinite sequences over a finite alphabet. The original definition in terms of
  gambling was proven robust by several characterizations in terms of  information lossless compression \cite{FSD}\ and several versions of entropy \cite{BoHiVi05}.

  Here we present an alternative definition on the Euclidean space and then prove its equivalence with \cite{FSD}.

  \begin{definition}
  Let $T$ be a $\Sigma$-FST and let $w\in\Sigma^*$. The {\sl $T$-information
  content of $w$\/} is  \[\K^T(w)=\min\myset{|\pi|}{T(\pi)=w}.\]
  \end{definition}

  \begin{definition}
  Let $T$ be a $\Sigma_b$-FST, $\delta>0$ and $x\in[0,1)$.  The {\sl base-$b$ $T$-information
  content of $x$ at precision $\delta$\/} is  \[\K^T_{\delta}(x)=\min\myset{\K^T(w)}{|\real_b(w)-x|<\delta}.\]
  \end{definition}

  We next define the finite-state dimension of points and sets.

    \begin{definition}
  Let $b\ge 1$. Let  $x\in[0,1)$ and $A\subseteq[0,1)$.  The {\sl base-$b$ finite-state dimension of $x$\/}
  is
  \[\dimfs^b(x)=\inf_{T \Sigma_b-\FST}\liminf_{\delta>0}\frac{\K^T_{\delta}(x)}{\log_b(1/\delta)},\]
  the {\sl base-$b$ finite-state dimension of $A$\/} is
  \[\dimfs^b(A)=\inf_{T \Sigma_b-\FST}\sup_{x\in A}\liminf_{\delta>0}\frac{\K^T_{\delta}(x)}{\log_b(1/\delta)}.\]
  \end{definition}

\begin{observation}
$\dimfs^b(x)=\inf_{T
\Sigma_b-\FST}\liminf_{n}\frac{\K^T_{b^{-n}}(x)}{n}$.
\end{observation}

The definition of finite-state dimension from \cite{FSD}\ is usually
done in a space of infinite sequences, while identifying $[0,1)$ and
$\Sigma_b^{\infty}$ through $\seq_b$ or base-$b$ representation.

Doty and Moser \cite{DotMos06}\ proved that finite dimension on
sequences can be characterized in terms of finite-state transducers.

\begin{theorem}[\cite{DotMos06}]
Let $S\in\Sigma^{\infty}$,
\[\dimfs(S)=\inf_{T \Sigma-\FST}\liminf_{n}\frac{\K^T(S\upto
n)}{n}.\]
\end{theorem}

We next show that the notion of information
  content  at a certain precision characterizes finite-state
  dimension.

\begin{theorem}For each $b\ge 1$, $x\in[0,1)$, and $A\subseteq[0,1)$
\[\dimfs^b(x)=\dimfs(\seq_b(x)),\]
\[\dimfs^b(A)=\dimfs(\seq_b(A)).\]
\end{theorem}

\begin{proof}
Let $x\in[0,1)$, let $S=\seq_b(x)$. Then for every $n\in \N$ and $T$
$\Sigma_b$-FST, $\K^T_{b^{-n}}(x)\le \K^T(S\upto (n+1))$ and
therefore $\dimfs^b(x)\le \dimfs(S)$.

For each $w\in \Sigma_b^{*}\cup\Sigma_b^{\infty}$, let $comp(w)$ be
the complementary of $w$, that is, $comp(w)[i]=b-1-w[i]$ for $0\le
i<|w|$.

\begin{claim}
$\dimfs(S)=\dimfs(comp(S))$.
$\dimfs^b(x)=\dimfs^b(\real_b(comp(\seq_b(x))))$.
\end{claim}

\begin{claim}
$\dimfs(S)\le \dimfs^b(x)$.
\end{claim}

To prove this claim, notice that $\dimfs^b(x)$ needs to be witnessed
either by approximations from above or for approximations from below
and that tighter approximations can be delayed.

That is, for every FST $T$ there exist and infinitely many $n_i$
such that
\[\lim_{i}\frac{\K^T_{b^{-n_i}}(x)}{n_i}= \liminf_n
\frac{\K^T_{b^{-n}}(x)}{n}.\] For each $i$ let $w_i$ be such that
$\K^T(w_i)=\K^T_{b^{-n_i}}(x)$ and $|x-\real_b(w_i)|<b^{-n_i}$. Let
$m_i$ be such that $b^{-m_i-1}\le |x-\real_b(w_i)|<b^{-m_i}\le
b^{-n_i}$. Then $\frac{\K^T_{b^{-m_i}}(x)}{m_i}\le
\frac{\K^T_{b^{-n_i}}(x)}{n_i}$,
 and for $T'(\pi)=comp(T(\pi))$, either
\[\liminf_{n}\frac{\K^T(S\upto n)}{n}\le \lim_i\frac{\K^T_{b^{-m_i}}(x)}{m_i}\] or
\[\liminf_{n}\frac{\K^{T'}(comp(S)\upto n)}{n}\le \lim_i\frac{\K^{T}_{b^{-m_i}}(x)}{m_i}.\]

Therefore either $\dimfs(S)\le \dimfs^b(x)$ or $\dimfs(comp(S))\le
\dimfs^b(x)$ and the claim follows.

\end{proof}

Since finite-state dimension in the space of sequences characterizes Borel normality
\cite{BoHiVi05}, we have an alternative characterization of normality in terms of finite-state dimension in the Euclidean space.

\begin{corollary}
Let $b\ge 1$, $x\in[0,1)$. $x$ is $b$-normal if and only if
$\dimfs^b(x)=1$, that is,
\[\inf_{T
\Sigma_b-\FST}\liminf_{\delta>0}\frac{\K^T_{\delta}(x)}{\log_b(1/\delta)}=1.\]
\end{corollary}

\section{Point to set principle for finite-state dimension}

We denote as {\sl separator\/} a set $S\subseteq [0,1)$ such that
$S$ is countable and dense in [0,1].

\begin{definition}
A {\sl separator enumerator (SE)\/} is a function $f:\Sigma^{*}\to
[0,1)$ such that $\mathrm{Im}(f)$ is a separator.
\end{definition}

For each  separator enumerator $f$ we can define information content in
$[0,1)$ relative to $f$.

 \begin{definition}
 Let $f:\Sigma^{*}\to
[0,1)$ be a SE.  Let $T$ be a $\Sigma$-FST, $\delta>0$ and
$x\in[0,1)$.  The {\sl $f$-$T$-information
  content of $x$ at precision $\delta$\/} is  \[\K^{T,f}_{\delta}(x)=\min\myset{\K^T(w)}{|f(w)-x|<\delta}.\]
  \end{definition}

   \begin{definition}
  Let $f:\Sigma^{*}\to
[0,1)$ be a SE. Let  $x\in[0,1)$ and $A\subseteq[0,1)$.  The {\sl
$f$-enumerator finite-state dimension of $x$\/}
  is
  \[\dimfs^f(x)=\inf_{T \Sigma-\FST}\liminf_{\delta>0}\frac{\K^{T,f}_{\delta}(x)}{\log_{|\Sigma|}(1/\delta)},\]
  the {\sl $f$-enumerator finite-state dimension of $A$\/} is
  \[\dimfs^f(A)=\inf_{T \Sigma-\FST}\sup_{x\in A}\liminf_{\delta>0}\frac{\K^{T,f}_{\delta}(x)}{\log_{|\Sigma|}(1/\delta)}.\]
  \end{definition}

We can generalize Borel normality through the same relativization.

\begin{definition}
Let $f:\Sigma^{*}\to [0,1)$ be a SE, let  $x\in[0,1)$. $x$ is {\sl
$f$-normal\/} if $\dimfs^f(x)=1$.
\end{definition}

  Given this natural relativization of finite-state dimension we next prove a point-to-set principle stating that for every set $A$ there exists an SE $f$ such that classical Hausdorff dimension of $A$ is exactly $f$-finite-state dimension. This implies that classical geometrical measure theory results can be obtained using only finite-state dimension.

\begin{theorem}
Let $A\subseteq[0,1)$.
\[\dimh(A)=\min_{f \mathrm{SE}}\dimfs^f(A),\]
\[\dimh(A)=\min_{f:\{0, 1\}^{*}\to \D_2 }\dimfs^f(A).\]

\end{theorem}

\begin{proof}
Let $C$ be such that $\dimh(A)= \dim^C(A)$ from the point-to-set
principle in \cite{LuLu18}. Le $U$ be the universal oracle Turing
Machine used in the definition of Kolmogorov Complexity for
effective dimension $\dim$. Let $h: \{0, 1\}^{*}\to \{0, 1\}^{*}$ be
such that $h(w)=U^C(w)$ when $U^C(w)$ is defined, and $U^C(w)=0$
otherwise. Then $f(w)=\real_2(h(w))$ is the required SE.
\end{proof}

Notice that the previous theorem holds even when fixing a particular
countable dense set. In terms of Borel normality, it shows that
reordering the set $\D_b$ of base-$b$ finite representation numbers
is enough to obtain normality for any other base.

\section{Conclusions and open questions}
We expect that our main theorem will prove new lower
bounds on Hausdorff dimension in different settings. Notice that the
result can be directly translated into any separable metric space and
any reasonable gauge family.

Our result helps clarify the oracle role in the point to set
principles. The next step would be to classify the different
enumerations of a countable dense set.

We believe that the notion of $f$-normal sequence can be of independent
interest with robustness properties inherited from those of the original concept, for instance from the fact that $x$ is $b$-normal exactly when the sequence $(b^nx)_n$ is equidistributed modulo 1.

\begin{Open}
    Let $f:\Sigma^{*}\to [0,1)$ be a SE, let  $x\in[0,1)$. For each $n\in\N$, let $a_n(x)= f(w)$ for $|w|\le n$ such that $f(w)\le x$ and  $x-f(w)$ is minimum.
    Can we characterize $f$-normality in terms of the equidistribution properties of $(|\Sigma|^n a_n(x))$?
    \end{Open}

\bibliographystyle{abbrv}
\bibliography{todo}

  \end{document}

%% file: commands.tex
\usepackage{amsfonts}
\usepackage{amsmath}
\usepackage{amssymb}
\usepackage{enumerate}

\newcommand{\upto}{\upharpoonright}

\DeclareMathOperator{\real}{real}

\DeclareMathOperator{\seq}{seq}

\newcommand{\FST}{\mathrm{FST}}

\newcommand{\myset}[2]{ \left\{ #1 \left| \, #2 \right. \right\} }

\newcommand{\ignore}[1]{}

\newcommand{\N}{\mathbb{N}}

\newcommand{\dimfs}{\mathrm{dim}_\mathrm{FS}}

\newcommand{\dimh}{\mathrm{dim}_\mathrm{H}}

\renewcommand{\dim}{{\mathrm{dim}}}

\newcommand{\D}{{\mathcal{D}}}

\newcommand{\K}{{\mathrm{K}}}



%% file: numbering2.tex
\usepackage{ifthen}

\newtheorem{theorem}{Theorem}[section]

\newtheorem{corollary}[theorem]{Corollary}
\newtheorem{claim}[theorem]{Claim}

\newenvironment{definition}
{ {\noindent {\bf Definition.}} } {  }

\newenvironment{example*}[1]
{ {\noindent {\bf Example #1.}} } {  }
\newenvironment{claim*}[1]
{ {\noindent {\bf Claim #1.}} } {  }
\newenvironment{theorem*}
{ {\noindent {\bf Theorem.}} } {  }
\newtheorem{observation}[theorem]{Observation}
\newenvironment{Open}
{ {\noindent {\bf Open question.}} } {  }

\newenvironment{proof}[1][xyzzy]
{
{\noindent {\bf Proof}%
\ifthenelse{\equal{#1}{xyzzy}}{{\bf .}}{~(#1).}} } { \hfill $\Box$ }

